\documentclass[12pt]{iopart}
\usepackage{amsmath}
\usepackage{amssymb}
\usepackage{amsthm}
\usepackage{url}
\usepackage{graphicx}
\newtheorem{fact}{Lemma}

\begin{document}

\title{Collapse in the nonlocal nonlinear Schr\"{o}dinger equation}
\author{F.~Maucher}
\address{Max Planck Institute for the Physics of Complex Systems, 01187, Dresden, Germany}
\ead{fabian@pks.mpg.de}
\author{W.~Krolikowski}
\address{Laser Physics Centre, Research School of Physics and Engineering,
Australian National University, Canberra, ACT 0200, Australia}
\author{S.~Skupin}
\address{Friedrich Schiller University, Institute of Condensed Matter Theory
and Optics, 07743 Jena, Germany}
\address{Max Planck Institute for the Physics of Complex Systems, 01187, Dresden, Germany}

\begin{abstract}
We discuss spatial dynamics and collapse scenarios of localized waves governed by the nonlinear Schr\"{o}dinger
equation with nonlocal nonlinearity.
Firstly, we prove that for arbitrary nonsingular attractive nonlocal nonlinear interaction in arbitrary dimension
collapse does not occur. Then we study in detail the effect of singular nonlocal kernels in arbitrary dimension using
both, Lyapunoff's method and virial identities.
We find that in the one-dimensional case, i.e. for $n=1$, collapse cannot happen for nonlocal nonlinearity. On the other hand, for spatial dimension $n\geq2$ and singular kernel $\sim 1/r^\alpha$, no collapse takes place if  $\alpha<2$, whereas collapse is possible if  $\alpha\ge2$. 
Self-similar solutions allow us to find an expression for the critical distance (or time) at which collapse should occur in the particular case of $\sim 1/r^2$ kernels.
Moreover, different evolution scenarios for the three dimensional physically relevant case of Bose-Einstein condensates 
are studied numerically for both, the groundstate soliton and higher order toroidal states with and without an additional local repulsive nonlinear interaction. In particular, we  show that presence of local repulsive nonlinearity can prevent collapse in those cases.
\end{abstract}
\pacs{42.65.Jx,03.75.Kk,52.35.Mw}

\section{Introduction}

 Spatial and/or temporal dynamics of waves in many nonlinear systems are often governed by the famous nonlinear Schr{\"o}dinger  (NLS) equation~\cite{Sulem:1999}. This universal equation appears in many diverse physical settings including, for instance, nonlinear optics~\cite{Agrawal:2006}, Bose-Einstein condensate~\cite{BEC_review, BEC_review2,Dalfovo:1999},   water waves~\cite{Zakharov:1968} and is even discussed in the context of rogue waves~\cite{Akhmediev:2009}.   A particular solution to this equation - a soliton, represents a nonlinear localized bound state
 that does not change its shape when propagating or evolving in time and/or space~\cite{Agrawal_Kivshar:2003}. The existence of solitons  is brought about by the balance between the natural effect of  dispersion or diffraction which tends to spread the wave and the self-focusing, provided by  the  nonlinear response  of the
 medium induced by the wave itself. In case of optics the nonlinearity  represents the self-induced refractive  index change, while, e.g. in Bose-Einstein-Condensates (BEC) it is just the nonlinear potential of bosonic interaction. In both cases the presence of the nonlinearity may lead to self-focusing and consequently, self-trapping  of light or condensate. When this nonlinear potential induced by the wave itself is strong enough  for the self-focusing to dominate over  the effects of dispersion/diffraction, the amplitude of the wave may grow to infinity (blow-up) within finite time or distance. This phenomenon is often referred to as a collapse.
 Collapse is  a well-known effect associated with dynamics of finite waves in nonlinear systems. It  has been studied in several physical settings including, e.g. plasma waves and nonlinear optics~\cite{berge:PhysRep:1998}, BEC~\cite{Hulet:2000}, as well as collapse  of very heavy stars into black holes~\cite{blackholes}. In the majority of  works devoted to collapse the nonlinear potential induced by the wave is of a local character, i.e. the system response in the particular spatial location is determined solely by the wave amplitude (or intensity) in the same location.

 Here we will discuss the dynamics of waves and collapse  in systems exhibiting a spatially  nonlocal nonlinear response. 
 In nonlocal media the nonlinear response
 at the given point depends on the wave intensity in an extended
 neighborhood of this point. The extent of this neighborhood in comparison to the
 spatial scale of the localized wave determines the degree of nonlocality \cite{wkob01}.
 It appears that spatial nonlocality is a generic property of many nonlinear systems
 and  is  often a result of an underlying transport process, such as
ballistic transport of atoms \cite{Saffman:07} or  atomic diffusion \cite{Suter:1993} in atomic vapors, heat transfer in plasma \cite{Litvak:1975:sjpp}  and thermal media \cite{Kleiman73,PhysRevLett.99.043903}, or
drift of electric charges  as in photorefractive crystals \cite{Zozulya96}. 
It can also be induced by  a long-range particle
 interaction as in  nematic liquid crystals  which exhibit orientational
 nonlocal nonlinearity \cite{Assanto:2003:jqe} or in colloidal suspensions \cite{PhysRevLett.97.123903}.
 Recent enhanced interest in nonlocal nonlinearities has been stimulated by the research on  
 Bose Einstein condensates   where the  inter-particle nonlinear interaction 
 potential is inherently nonlocal \cite{Dalfovo:1999,ParSalRea98,Lahaye:09:prp}.

In contrast to local attractive nonlinearity, which leads to collapse in more than one dimension, nonlocal nonlinearities may actually prevent collapse.
 Turitsyn was the first to prove analytically the absence of collapse in the nonlocal NLS equation for three different singular nonlocal kernels in \cite{Turitsyn:85:tmf}, including  the singular gravity-like $1/r$-kernel.  In a later work Perez-Garcia {\em et al.}  investigated the collapse dynamics of the  nonlocal Schr\"{o}dinger equation using variational techniques and numerical simulations \cite{Perez:pre:62:4300}. They showed in particular, that for the nonsingular nonlocal nonlinearity  collapse is always arrested.
Absence of collapse for a singular $1/r$-response and form stability of higher order states, such as dipole and torus, have been recently verified numerically in \cite{Maucher:2010:PRA}.
More general treatment of the NLS and ensuing dynamics and collapse is presented in \cite{Cazenave:2003,Froehlich:2003}.
As far as collapse in the general case of the nonlocal Schr\"{o}dinger equation is concerned, Bang {\em et al.} showed rigorously that  nonlocality arrests the collapse as long as the Fourier spectrum of the nonlocal nonlinear potential is absolute integrable \cite{Krolikowski:job:6:5}. 
Rigorous recent studies of collapse in three dimensions in case of a particular singular form  of the nonlocal nonlinear potential have been given in \cite{Lushnikov:PRA:051601,chen:physD:227:2,Lushnikov:PRA:023615}.
Here, we will generalize those previous findings to arbitrary number of transverse dimensions, and show that for the nonlocal nonlinearity represented by singular kernels $\sim 1/r^\alpha$, $\alpha=2$
there is a sharp threshold for collapse. By means of (3+1)-dimensional numerical simulations we elucidate different evolution scenarios, and show that an additional local repulsive nonlinear interaction
may prevent collapse.

\section{Model}

In this paper we will consider the self-trapped evolution of the wave function $\psi(x,t)$ described by the nonlocal nonlinear Schr{\"o}dinger equation (NNLS), given by
\begin{equation}\label{eq:nls}
i \partial_t \psi(x,t) = -\Delta \psi(x,t) - \left(\int A\left(x-y\right)
\left|\psi\left(y,t\right)\right|^2\mathrm{d}y \right) \psi(x,t),
\end{equation}
where $x,y\in \mathbf{R}^n$ with $n$ being the spatial dimension. $A(x)$ is the so called   nonlocal response or kernel which represents the nonlocal character of the nonlinearity of the medium. Its width determines the degree of nonlocality~\cite{wkob01}. In particular, for  $A(x)\sim\delta(x)$ the above equation describes the standard Kerr-type local medium. While this is only a phenomenological model it nevertheless describes very well general properties of nonlocal media. Moreover, for certain forms of the nonlocal response function this model represents actual experimentally realizable physical systems.
This is the case of  nematic liquid crystals~\cite{Assanto:2003:jqe, Conti:2003:prl} or thermal nonlinearity in liquids~\cite{PhysRevLett.99.043903} and plasma~\cite{Litvak:1975:sjpp} where, in one dimensional case $A(x)\sim \mathrm{exp}(-|x|)$ and in two dimensions $A(x)$ is proportional to the modified Bessel function of the second kind~\cite{Skupin:pre:73:066603}.
In addition,  Eq.~(\ref{eq:nls}) (also known as  Gross-Pitaevskii equation) governs the spatio-temporal dynamics of  Bose-Einstein condensates where the analytical form of the nonlocal Kernel is a consequence of the dipolar character of the   long-range particle-particle interaction \cite{Lahaye:09:prp,Goral:05,Griesmaier:05,Pedri:05}. While, in general, the nonlocal interaction potential in condensates is of rather complex and anisotropic character, under certain experimental conditions it may acquire ``gravity-like'' $1/r$ form~\cite{Odell:00}.

\section{Analytical estimates for possible evolution scenarios}

In a vivid picture, collapse refers to a phenomenon when the amplitude of the wave function $\psi$ goes to infinity within finite time or distance (for stationary processes).  However, in the mathematical context collapse is usually defined by means of a diverging gradient norm $\|\nabla \psi\|_2$  (see \ref{notations} for notations). For practical considerations, this is essentially the same, because a wave function maintaining a
finite $L^2$ norm, $\|\psi\|_2$, attains infinite amplitude when its gradient norm $\|\nabla \psi\|_2$ diverges~\cite{Glassey:1977}. 

\subsection{Lyapunoff methods}
In this subsection, we will use Lyapunoff's method \cite{Salle:1961}, to discuss  collapse scenarios. The idea is, roughly speaking, to show that the Hamiltonian
\begin{equation}
\mathcal{H} = \iint \left| \nabla \psi \right|^2 \mathrm{d}x\mathrm{d}y
- \frac{1}{2}\iint \left|\psi\left(x\right)\right|^2  \\
 A\left(x-y\right)
\left|\psi\left(y\right)\right|^2\mathrm{d}x\mathrm{d}y
\end{equation}
for Eq.~(\ref{eq:nls}) is bounded from below. Then, since in our case the $L^2$ norm $\|\psi\|_2$ and the Hamiltonian $\mathcal{H}$ are conserved quantities, it is possible to show that
the gradient norm $\|\nabla \psi\|_2$ is bounded from above, which implies that it cannot diverge.
Moreover, we can expect the existence of at least one stable soliton solution when the Hamiltonian is bounded~(see, e.g., \cite{Kuznetsov:1986}).

We can decompose a general kernel $A=A^+ + A^-$, $A^+\geq0$ and $A^-\leq 0$. Then, the  kernels $A^+$ and $A^-$ represent  attraction (in BEC) or focusing (in nonlinear optics) and repulsion or defocussing, respectively.
Since the amplitude of the wave function has to become infinity at a certain point, one can expect that a singular kernel $A^+$ promotes collapse.
Using Young's convolution theorem, one can immediately infer the following  related fact.
Without loss of generality, we restrict our following considerations to solely attractive kernels $A=A^+$.

\begin{fact}
Let $A\in L^\infty(\mathbf{R}^n)$, $\psi(t=0,x)\in H^1(\mathbf{R}^n)$. Then no collapse occurs at all times, since $\|\nabla \psi\|_2^2 $ is bounded from above.
\end{fact}
\begin{proof}
\begin{equation}\label{eq:nonsingular_estimate}
\mathcal{H} \geq \| \nabla \psi \|_2^2-\frac{1}{2}C \|\psi\|_2^4  \|A\|_\infty.
\end{equation}
Here, $C=1$ is a sharp constant (see, e.g., \cite{Lieb:2001}).
Since $\|\psi\|_2$ and the Hamilton $\mathcal H$ itself are conserved quantities under the flow of time (or propagation), one can read the inequality (\ref{eq:nonsingular_estimate})
as bounding $\| \nabla \psi \|_2^2$ from above in terms of conserved quantities.
\end{proof}

Let us try to connect the inequality~(\ref{eq:nonsingular_estimate}) with the earlier results of the stability of finite beams governed by the NNLS equation~(\ref{eq:nls})~\cite{Krolikowski:job:6:5}. This work states that in order to arrest collapse the spectrum of the nonlocal kernel $A$  should be
absolutely integrable. Since $\| A \|_{\infty} \leq C \|\tilde{A}\|_1$, where $\tilde{A}$ denotes the Fourier transform of $A$ (see \ref{notations}), the inequality~(\ref{eq:nonsingular_estimate}) clearly contains the result of Ref.~\cite{Krolikowski:job:6:5} and also explains earlier reports on collapse arrest for nonsingular kernels
 observed in numerical simulations~\cite{Perez:pre:62:4300}.

Further on, inequality~(\ref{eq:nonsingular_estimate}) shows that collapse can only be relevant if the kernel $A$ diverges
at a certain point (usually at the origin $x=0$). Then the assumption $A\in L^\infty(\mathbf{R}^n)$
is violated, as is the case, e.g., of local interaction [$A\sim\delta (x)$].
However, a diverging kernel $A$ does not necessarily imply collapse.
In order to ease the restrictions on the nonlocal kernel which ensure the collapse arrest we
have to restore to  the gradient norm $\|\nabla \psi\|_2$ apart from the $L^2$-norm $\|\psi\|_2$.
As we will show below one can weaken the sufficient restriction $A\in L^{\infty}$  by employing  the Hardy-Littlewood-Sobolev inequality (see \ref{hardy_littlewood}) and  the Gagliardo-Nirenberg-Sobolev inequality (see \ref{gagliardo_nirenberg}).

\begin{fact}
Let $x\in\mathbf{R}^n$, $A = f(x)/r^\alpha$, $r=|x|$ with $0\leq f(x)< \infty$, $\alpha<\mathrm{min}(2,n)$, $\psi(t=0,x)\in H^1(\mathbf{R}^n)$, $\kappa:=\mathrm{max_x f}<\infty$. Then collapse cannot occur.
\label{fact2}
\end{fact}
\begin{proof}
\begin{equation}\label{eq:estimate2}
\begin{split}
 \mathcal{H} &  \geq  \| \nabla \psi \|_2^2 - \frac{1}{2} \kappa \iint \frac{|\psi(x)|^2|\psi(y)|^2}{|x-y|^\alpha}\mathrm{d}y \mathrm{d}x \\
 & \geq  \| \nabla \psi \|_2^2 - \frac{1}{2} \kappa C\| \psi \|_{2n/(n-\alpha/2)}^4 = \| \nabla \psi \|_2^2 - \frac{1}{2} \kappa C \left(\int |\psi|^{2n/(n-\alpha/2)} \mathrm{d}x\right)^{4\frac{n-\alpha/2}{2n}}
\end{split}
\end{equation}
In the second step we just applied the Hardy-Littlewood-Sobolev inequality (\ref{eq:hardy_littlewood}).
Introducing $\sigma = \alpha / (2n-\alpha)$, we can rewrite the last  expression in~(\ref{eq:estimate2}) and find:
\begin{equation}
\label{eq:estimate3}
\begin{split}
\mathcal{H}&\geq \| \nabla \psi \|_2^2 - \frac{1}{2} \kappa C \left(\int |\psi|^{2n/(n-\alpha/2)} \mathrm{d}x\right)^{4\frac{n-\alpha/2}{2n}} =\| \nabla \psi \|_2^2 - \frac{1}{2} \kappa C \left(\int |\psi|^{2+2\sigma} \mathrm{d}x\right)^{\frac{\alpha}{n\sigma}} \\
&\geq \| \nabla \psi \|_2^2 - \frac{1}{2} \kappa C \left[K\left(\int |\nabla \psi|^2 \mathrm{d}x\right)^{\frac{\sigma n}{2}} \left(\int|\psi|^{2}\mathrm{d}x\right)^\frac{2+\sigma
(2-n)}{2}\right]^{\frac{\alpha}{n\sigma}} \\
& =  \| \nabla \psi \|_2^2 - \frac{1}{2} \tilde{C} \| \nabla \psi \|_2^{\alpha} \| \psi \|_2^{4-\alpha}.
\end{split}
\end{equation}
Here, we employed the  Gagliardo-Nirenberg-Sobolev inequality (\ref{eq:GN}), and $\tilde{C}:=\kappa C K^{\alpha/n\sigma}$ is some positive constant independent of $\psi$.
Since $\alpha<2$ and because the Hamiltonian $\mathcal{H}$ and the $L^2$-norm $\|\psi \|_2$ are conserved quantities, inequality~(\ref{eq:estimate3})
can be read as a uniform bound of the gradient norm $\| \nabla \psi \|_2$.
\end{proof}

For dimension $n=1$, the convolution term in the Hamiltonian $\mathcal{H}$ and expression~(\ref{eq:estimate2}) diverge without regularization for singular kernels with $\alpha\geq1$, and
hence singularity cannot be strong enough to promote collapse.
Generally, when regularizing the kernel by, e.g., introducing a cut-off,
again the wave function does not collapse since the kernel is then nonsingular.

If we apply the lower bound for the Hamiltonian (\ref{eq:estimate3}) for the particular exponent $\alpha=2$ we find
\begin{equation}
\label{eq:critical_mass}
\mathcal{H} \geq  \| \nabla \psi \|_2^2 \left(1 - \frac{1}{2} \tilde{C}  \| \psi \|_2^2\right).
\end{equation}
Because we know the sharp constant $\tilde{C}:=\kappa C K^{\alpha/n\sigma}$ [$\kappa=1$, $C=C(n,\alpha)$ according to Eq.~(\ref{eq:hardy_littlewood_constant}), $K=K(\sigma,n)$ according to Eq.~(\ref{eq:GN_constant})] we can use the inequality (\ref{eq:critical_mass}) to compute a lower bound for the ''critical mass'' $\| \psi \|_2^2=2/\tilde{C}$ for collapse. In the physical relevant case of $n=3$, we find this lower bound numerically as $2/\tilde{C}=2.80$~\footnote{Note that the constant $K=K(\sigma,n)$ involves the quantity $\|R\|_2^{2\sigma}$, which has to be computed numerically.}, which is slightly lower than the numerical result for the mass of the groundstate soliton, $\| \psi_{\rm GS} \|_2^2=2.85$ (see Sec.~\ref{gs_num}). 
This finding may seem surprising, because for the local nonlinear Schr\"odinger equation the exact soliton mass can be computed from sharp constants in the lower bound of the Hamiltonian~\cite{Weinstein:1983,Turitsyn:1993}.
However, a closer look at the Hardy-Littlewood-Sobolev inequality (Sec.~\ref{hardy_littlewood}) reveals that equality can be expected only for wavefunctions $\psi$ of the form (\ref{eq:hardy_littlewood_equation}). Since our groundstate solitons $\psi_{\rm GS}(x)$ are not of the form~(\ref{eq:hardy_littlewood_equation})~\footnote{$\lim\limits_{r \rightarrow \infty}\psi_{\rm GS}\sim\exp(-r) \nsim r^{-4}$, with $|x|=r$.}, we have to expect $2/\tilde{C}<\| \psi_{\rm GS} \|_2^2$. Finally, we verified numerically that for the groundstate soliton (\ref{eq:hardy_littlewood}) is indeed an inequality. Interestingly, then the ratio of r.h.s. and l.h.s. of (\ref{eq:hardy_littlewood}) is about $2.85/2.80=\| \psi_{\rm GS} \|_2^2 \times \tilde{C}/2$.

\subsection{Virial identities}

The idea to use virial identities to derive sufficient conditions for collapse was first established in the context of the local two-dimensional nonlinear Schr\"odinger equation by Vlasov, Petrishchev, and
Talanov~\cite{vpt}, and the resulting condition $\mathcal{H}<0$ and is often referred to as the VPT criterion. The original idea behind the VPT criterion has been extended to various modifications of the original NLS equation, e.g., the Gross-Pitaevskii equation for a BEC with dipole-dipole interactions~\cite{Lushnikov:PRA:051601} or, more general, the NLS with $1/r^{\alpha}$ nonlocal response in three dimensions~\cite{Lushnikov:PRA:023615}.
Following the calculations presented in \cite{Lushnikov:PRA:023615}, one may observe that the derived virial identities hold for arbitrary dimension.
The virial is given by $V(t):=\int r^2 |\psi|^2 \mathrm{d}x$ with $r=|x|$.

\begin{fact}
Consider a kernel $A= 1/r^\alpha$, with $x\in\mathbf{R}^n$ and $\alpha<n$.
Then the second time derivative of the virial in arbitrary dimension $n$ is given by
\begin{equation}
\begin{split}
\partial_{tt} V & = 8 \left( \int |\nabla \psi|^2 \mathrm{d}x - \frac{\alpha}{4}\iint \frac{|\psi(y)|^2|\psi(x)|^2}{|x-y|^\alpha} \mathrm{d}y\mathrm{d}x\right) \\
 & =8\mathcal{H} -2\left(\alpha-2\right)\iint \frac{|\psi(y)|^2|\psi(x)|^2}{|x-y|^\alpha} \mathrm{d}y\mathrm{d}x
\end{split}
\end{equation}
\label{fact3}
\end{fact}

\begin{proof}
This task was already done for three dimensions in \cite{Lushnikov:PRA:023615}.
Therefore we only briefly sketch the derivation of this virial identity here.
Apparently,
\begin{equation}
\partial_t V = 4\Im\int \psi^* x\nabla \psi \mathrm{d} x
\end{equation}
Using the equation of motion, a longer but straightforward calculation leads to (see also \cite{chen:physD:227:2})
\begin{equation}
\label{intermediate}
\partial_{tt} V = 8\int |\nabla \psi|^2- 4n \int \frac{|\psi(x)|^2|\psi(y)|^2}{|x-y|^\alpha} \mathrm{d}y \mathrm{d} x \\
+4\int x\nabla (|\psi (x)|^2) \int \frac{|\psi(y)|^2}{|x-y|^\alpha}\mathrm{d}y\mathrm{d}x
\end{equation}
Integrating the last term by parts and re-writing it in a more favorable way gives
\begin{align}
&\quad\int x\nabla (|\psi (x)|^2) \int \frac{|\psi(y)|^2}{|x-y|^\alpha}\mathrm{d}y\mathrm{d}x \\
&= n\iint \frac{|\psi (x)|^2 |\psi(y)|^2}{|x-y|^\alpha}\mathrm{d}y \mathrm{d}x -\frac{\alpha}{2}\iint x\frac{ (x-y)|\psi (x)|^2 |\psi(y)|^2}{|x-y|^{\alpha+2}}\mathrm{d}y \mathrm{d}x\\
&-\frac{\alpha}{2}\iint y\frac{ (y-x)|\psi (x)|^2 |\psi(y)|^2}{|x-y|^{\alpha+2}}\mathrm{d}y \mathrm{d}x. \nonumber
\end{align}
Plugging this into Eq.~(\ref{intermediate}), the first term cancels out, and the two symmetric terms with respect to $x$ and $y$ add up and give the desired expression.
\end{proof}

In particular, Lemma \ref{fact3}
implies $\partial_{tt}V=8\mathcal{H}$ for $A=1/r^2$ (or $\alpha=2$) in arbitrary dimension.
Together with a particular form of H{\"o}lder's inequality
\begin{equation}
\label{hoelder}
 \int r |\psi|  |\nabla \psi| \mathrm{d}x \leq  \left(\int r^2 |\psi|^2 \mathrm{d}x \right)^{1/2}  \left(\int |\nabla \psi|^2 \mathrm{d}x\right)^{1/2},
\end{equation}
we can derive a sufficient criterion for collapse.

\begin{fact}
Let $x\in\mathbf{R}^n$, $n>2$, $A = 1/r^2$, $\psi(t=0,x)\in H^1(\mathbf{R}^n)$, and the Hamiltonian $\mathcal{H}<0$. Then the solution $\psi$ undergoes collapse.
\label{fact4}
\end{fact}
\begin{proof}
As seen above, the second time derivative of the virial $V$ equals eight times the Hamiltonian. The virial is positive for $t=0$ and has the apparent meaning of the mean width of $\psi$. Hence, for a negative Hamiltonian, the virial vanishes at a certain time $t$.

Using partial integration and H{\"o}lder's inequality~(\ref{hoelder}),
one can infer that a vanishing virial leads to collapse due to conservation of the $L^2$-norm $\|\psi \|_2$ of the wavefunction and
\begin{equation}
\int |\psi|^2 \mathrm{d}x \leq C \left( \int r^2 |\psi|^2 \mathrm{d}x \right)^{1/2} \left( \int |\nabla \psi|^2 \mathrm{d}x \right)^{1/2},
\end{equation}
where the positive constant $C$ does not depend on $\psi$.
\end{proof}

In other words, if the mean width or the virial becomes zero, the gradient norm has to diverge in order to keep the $L^2$-norm constant. Closer inspection of Lemma~\ref{fact3} reveals that a negative Hamiltonian also implies collapse for $\alpha>2$, as already pointed out in~\cite{Lushnikov:PRA:023615}.
Together with Lemma~\ref{fact2}, we see that $\alpha=2$ is a sharp boundary for collapse in arbitrary dimension $n>2$.

\subsection{Self-similar solutions for the $1/r^2$ kernel}\label{sec:self-similar}

Let us now have a look at the special case $A = 1/r^2$. As we have seen in the previous section, this singular kernel
marks the boundary between collapsing and non-collapsing solutions. This is somewhat similar to the two-dimensional local
nonlinear Schr\"odinger equation, and we will find many aspects of this famous equation here. However, in the nonlocal case the collapse threshold
does not appear as a critical dimension, but as a critical exponent of the singularity of the kernel.

The NNLS equation with kernel $A = 1/r^2$ appears as equation of motion of the following Lagrangian density,
\begin{equation}
\mathcal{L}:=\frac{i}{2}\left(\psi\partial_{t}\psi^{*}-
\psi^{*}\partial_{t}\psi\right)+
\left|\nabla\psi\right|^{2}-
\frac{1}{2}\left|\psi\right|^{2}
\int \frac{1}{|x-y|^2}\left|\psi\left(y,t\right)\right|^{2}
\mathrm{d}y.
\end{equation}
To simplify the considerations and following \cite{berge:PhysRep:1998}, let us assume radial symmetry. A possible ansatz for a self-similar solution (or similariton \cite{Dudley:NatPhys:2007}) is
\begin{equation}
\psi = B(t) R(r/a(t)) e^{i\theta (r,t)}
\label{var_ansatz}
\end{equation}
with $r=|x|$. Here, $R(t/a(t))$ is a real profile, and $a(t)$ is the wave radius. Due to the conservation of the $L^2$-norm $\|\psi \|_2$ and the related continuum equation, which comes about  from the global $U(1)$ invariance of the Lagrangian, one finds the interdependence between $B$ and $a$:
\begin{equation}
 |B|^2 a^n =: P = const.
\end{equation}
 An appropriate ansatz for the radially symmetric phase is
\begin{equation}
\theta(r,t) = \theta_2 (t) r^2 +\theta_0(t).
\label{var_phase}
\end{equation}
 Using $\theta$ and $P$, one may express the spatially integrated Lagrangian density as follows:
\begin{equation}
L=-|B|^2 a^n \left[ \kappa_0 \partial_t\theta_0 + \kappa_2 a^2 \left( \partial_t\theta_2 +4\theta_2^2 \right)+ \frac{\gamma}{a^2}  \right] \\
-\beta |B|^4 a^{2n-2},
\end{equation}
where we introduced
\begin{equation}
\kappa_m :=  \int r^{m} R^2(r) \mathrm{d}x, \quad
\beta:= 1/2 \int \frac{R^2(|y|)R^2(r)}{|x-y|^2}\mathrm{d}x\mathrm{d}y \textrm{, and} \quad
\gamma := \int (\partial_r R(r))^2 \mathrm{d}x.
\end{equation}
The variational problem $\delta L=0$ with respect to $B$, $\theta_0$, $\theta_2$ and $a$ gives the equations of motion for these parameters.
Using $P=|B|^2 a^n$ and $\theta_2=(\partial_t a)/4a$, one may rewrite the equation of motion for $a$ in the form
\begin{equation}
\partial_{tt}a = \frac{4\gamma}{\kappa_2}\frac{1}{a^3} - \frac{4\beta P}{\kappa_2} \frac{1}{a^3}.
\end{equation}
Elementary integration with the initial conditions $a_0 := a(t=0)$, and $\partial_t a(t=0) = 0 $ gives
\begin{equation}
a(t) = a_0 \sqrt{ t^2 \left( \frac{4\gamma}{\kappa_2 a_0^4} - \frac{4\beta P}{\kappa_2 a_0^4}\right) +1 }.
\end{equation}
Now one can estimate the critical time $t_{\rm cr}$, where $a\rightarrow 0$ as
\begin{equation}\label{eq:t_c}
t_{\rm cr} = \left( \frac{4\beta P}{\kappa_2 a_0^4}   -\frac{4\gamma}{\kappa_2 a_0^4}\right)^{-1/2}.
\end{equation}
In terms of $P$, one may express the condition for collapse as $P>\frac{\gamma}{\beta}$. Compared to the virial results from the previous section,
this threshold value indicates that the Hamiltonian $\mathcal{H}$ becomes negative.

It is important to stress here that the variational approach presented above strongly relies on the ansatz Eqs.~(\ref{var_ansatz}) and~(\ref{var_phase}). A priori, there is no guarantee that variational results show even qualitative agreement with the exact solutions. A classical example of unsuccessful application of the variational method is given in~\cite{varfail}, where the authors show that the description of solitons interacting with radiation in the one dimensional local NLS equation fails. However, as far as collapse in the two and three dimensional local NLS is concerned, the above ansatz is capable of giving useful estimates~\cite{2dcollapse,3dcollapse,berge:PhysRep:1998}. As we will see in the following sections, in our case variational predictions show some qualitative agreement with rigorous numerical simulations in the physically relevant case of $n=3$. In the light of Ref.~\cite{3dcollapse}, this is indeed the best we can expect.

\section{(3+1)-dimensional numerical investigation of collapse dynamics and collapse arrest for the $1/r^2$ kernel}

In this section we will investigate the dynamics of finite wave packets using  numerical integration of the
NNLS equation~(\ref{eq:nls}).
  We start by using variational methods to find families of approximate localized solutions to the stationary version of this equation,
and compare them to their numerically exact counterparts computed by imaginary time evolution~\cite{Tosi} or iteration technique~\cite{Skupin:pre:73:066603}.
Both of them will be subsequently used as initial conditions in numerical simulations in order to illustrate possible scenarios of
time evolution of the wave function in case of the nonlocal $1/r^2$ kernel. Our (3+1)-dimensional simulation code is based on a Fourier split-step method.
In order to resolve the singularity of the kernel properly, which is of crucial importance when studying collapse phenomena, we use a method outlined
in \ref{sec:resolve_singularity}.
In the following, we will also investigate the impact of an additional local repulsive
(defocusing) term on the collapse dynamics. In particular, we will be interested
whether the presence of the repulsive local interaction can arrest the collapse.
To this end, let us consider a NNLS equation of the following form:
\begin{equation}
\label{nlsnumerics}
i \partial_t \psi(x,t) = -\Delta \psi(x,t) - \left[\int \frac{1}{|x-y|^2}\left|\psi\left(y,t\right)\right|^{2}
\mathrm{d}y-\delta\left|\psi(x,t)\right|^{2} \right] \psi(x,t).
\end{equation}
where  the parameter  $\delta$ in front of the local term is either $0$ or $1$.
On physical grounds, the additional local repulsive term occurs in Bose-Einstein-Condensates as a result of $s$-wave scattering~\cite{Dalfovo:1999}.

\subsection{Approximate variational and exact numerical solutions for the groundstate \label{gs_num}}

To find approximate localized groundstate solutions to the NNLS equation we use the following ansatz
\begin{equation}\label{gaussian}
\psi(x,t) := A_0\exp\left(-\frac{r^{2}}{2\sigma^{2}}\right)\exp(iEt)
\end{equation}
and vary the corresponding Lagrangian with respect to width $\sigma$ and amplitude $A_0$. Here, we restrict our considerations to dimension $n=3$, so that $x=(x_1,x_2,x_3)$.
The results are shown in Figs.~\ref{fig:1}(a-b), which depict the so called existence curves for the solutions, i.e., the relations between mass $M:=\|\psi \|_2^2$, energy $E$
 and width $\sigma$.

The solid lines in Fig.~\ref{fig:1}(a) represent families of solutions obtained by variational analysis with (black line) and without (blue line) repulsive local interaction ($\delta=0$),
and the dashed green horizontal line together with the blue dots show the same for the numerical exact soliton solution of Eq.~(\ref{nlsnumerics}). 
We expect from this plot, that the groundstate may indeed collapse without repulsion ($\delta=0$), since for fixed mass $M$ the width $\sigma$ can be arbitrarily small.
Typical length-, time- and amplitude-scales $R_c$, $T_c$ and $\psi_c$ for an exact solution behave like $1/T_c\sim 1/R_c^2\sim \psi_c^2 R_c$. Hence, the 
invariance of the mass $M$ with respect to energy $E$ can be expected from $M\sim\psi_c^2 R_c^3\sim \rm{const}$. The same argument holds for arbitrary dimension $n$,
where we have $1/R_c^2\sim \psi_c^2 R_c^n/R_c^2 $.
In contrast, when including the additional local repulsive interaction competing with the nonlocal attractive one, we expect absence of collapse from Fig.~\ref{fig:1}(a), since the
mass $M$ of the groundstate is increasing with decreasing width $\sigma$.

\begin{figure}
\includegraphics[width=0.99\textwidth]{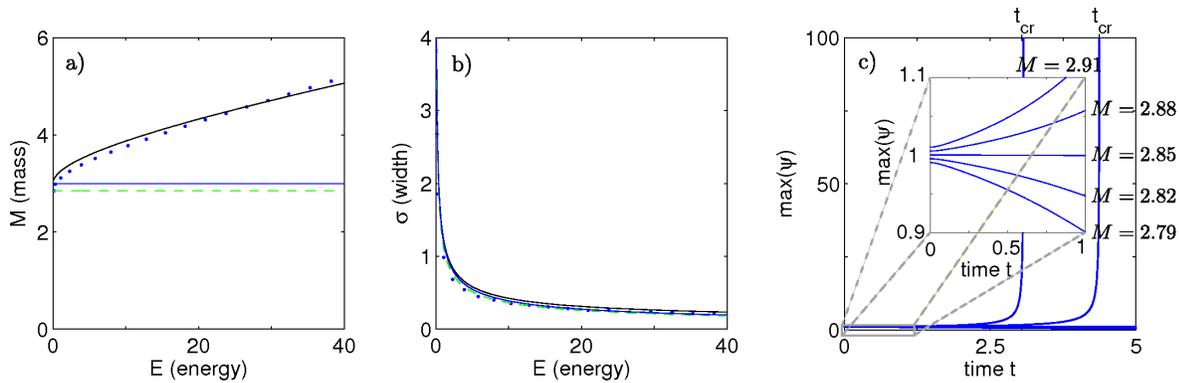}
\caption{Existence curves of the groundstate soliton [(a) and (b)] obtained by variational analysis and numerical computation for
nonlocal attractive $1/r^2$ response (blue solid and dashed green line, resp.) and with competing local repulsive interaction (black solid and dotted blue line, resp.).
c) Time evolution of the amplitude of the numerically exact groundstate. If the mass of the groundstate is increased, the solution collapses (shown here for $\sqrt{1.01}$ and $1.01 \times \psi_{\rm GS}$), and if the mass is decreased, the solution spreads upon propagation (shown here for $0.99$ and $\sqrt{0.99} \times \psi_{\rm GS}$). The inset shows a zoom into the early propagation stage.
}
\label{fig:1}
\end{figure}

Before presenting rigorous numerical simulations of the time evolution of the above solutions, let us
evaluate the expression for the critical collapse time $t_{\rm cr}$ derived in section \ref{sec:self-similar}. Plugging                                                                                                                                             
the Gaussian approximate variational groundstate as well as the numerically exact soliton as self-similar profile into Eq.~(\ref{eq:t_c}), the critical time $t_{\rm cr}$ diverges, and we expect that collapse happens at infinite times. However, when multiplying the profiles by some factor larger than one, $t_{\rm cr}$ becomes finite, and we expect collapse at time $t_{\rm cr}$.

These predictions are confirmed by our numerics for the exact groundstate, as can be seen in Fig.~\ref{fig:1}(c). The groundstate is self-trapped and evolves with small change in amplitude (invisible in the figure) due to numerical imperfections.
If we multiply the exact groundstate
by a factor smaller than one (e.g. $\sqrt{0.99}$), it spreads with time, whereas if we multiply it by a factor larger than one (e.g. $\sqrt{1.01}$), it collapses. The critical times $t_{\rm cr}$ obtained from Eq.~(\ref{eq:t_c}) are in excellent agreement with the simulations.
In contrast, using the approximate variational groundstate as initial condition, we see that it collapses [Fig.~\ref{fig:2}a)] after rather short time. Thus, the estimate for the critical time $t_{\rm cr}=\infty$ is incorrect. This is related to the fact
that the approximate variational Gaussian profile has a larger mass than the exact soliton ($M=3 > M _{\rm GS}$). 
When we multiply the variational groundstate by a factor larger than one [see Fig.~\ref{fig:2}(b)], we see collapse earlier than predicted by Eq.~(\ref{eq:t_c}). Collapse-times $t_{\rm cr}$ are systematically overestimated.
Finally, collapse can be prevented by adding a local repulsive term to the nonlinear potential in the original NNLS equation [see Fig.~\ref{fig:2}(c)].

 \begin{figure}
\includegraphics[width=0.99\textwidth]{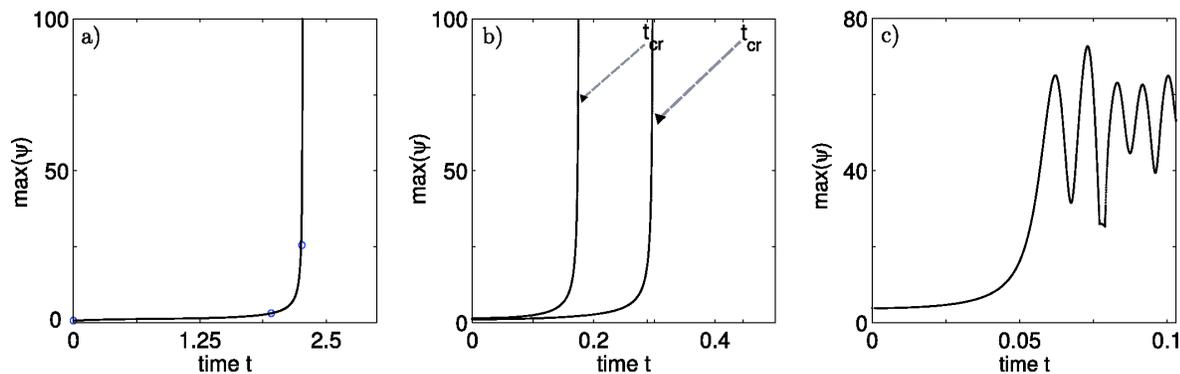}
 \caption{
The evolution of the amplitude for the initial variational  groundstate is shown in a). The blue circles correspond to snapshots of the evolution shown in Fig.~\ref{fig:3}.
In b), the evolution of the variational groundstate multiplied by $2$ (left curve) and $1.5$ (right curve)
is plotted and compared to the estimates for the critical time $t_{\rm cr}$.
In c), the possibility to prevent collapse using a local repulsive term is
illustrated: The variational groundstate multiplied by a factor of $5$ remains much smaller than the collapsing ones without repulsion shown in b). }
\label{fig:2}
 \end{figure}

The time evolution of the initially Gaussian variational groundstate in Fig.~\ref{fig:2}(a) can provide us with some deeper insight into the collapse dynamics. From our previous attempts to predict the critical collapse time $t_{\rm cr}$, we may suspect that the exact soliton itself is the self-similar profile. In order to corroborate this hypothesis, Figs.~\ref{fig:3}(a), (b) and (c)
show some snapshots of the radial profiles during time evolution corresponding to Fig.~\ref{fig:2}(a). 
Already after evolving over a short time, the wave function
clearly deviates from a Gaussian profile, and resembles more the numerically exact solution [see black curves in Fig.~\ref{fig:2}(b)]. Further propagation shows that
the collapsing solution converges to the exact groundstate. A similar behavior is know for the classical two-dimensional local nonlinear
Schr\"odinger equation, where the famous Townes soliton naturally emerges in the collapse process (see, e.g., ~\cite{fraiman,PhysRevLett.90.203902}).
Moreover, we checked that an initially ellipsoidal wave function with slightly too large mass will collapse as a radially symmetric formation, that again resembles the numerically exact solution (not shown).

\begin{figure}
\includegraphics[width=0.99\textwidth]{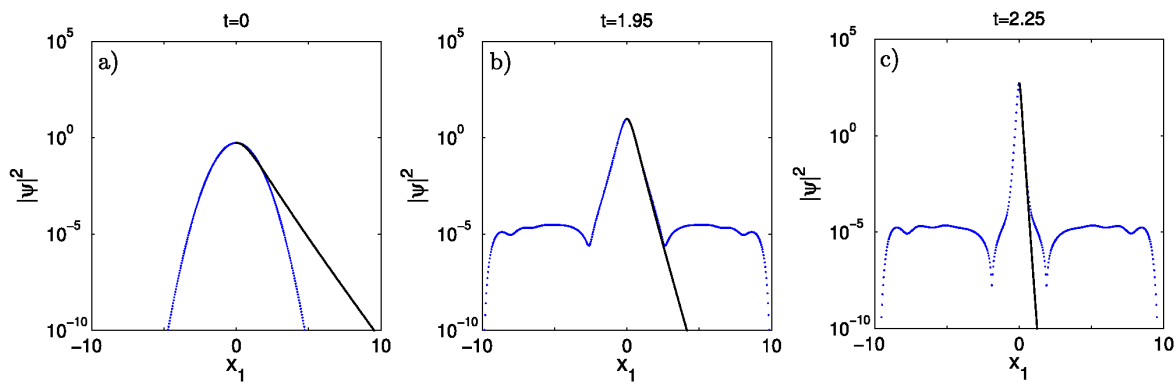}
\caption{
Snapshots of the time evolution of the initially Gaussian variational groundstate of Fig.~\ref{fig:2}(a) at a) $t=0$, b) $t=1.95$ and c) $t=2.25$. The black curves show numerically exact solitons with the same amplitude. Clearly, some part of the initially Gaussian wave function is radiated away, and the actually collapsing part 
converges towards the exact soliton.}
\label{fig:3}
\end{figure}

\subsection{Evolution scenarios for toroidal states}

Finally, we want to have a look at higher order solitons and their collapse behavior. For this purpose, we consider toroidal states, i.e. the three-dimensional generalization of a vortex. Toroidal states have an azimuthally symmetric density together 
with an azimuthal phase ramp, which leads to a phase singularity of integer topological charge $m$ in the origin
\begin{equation}
\psi(r,\phi,z,t) = \rho(r,z)\exp(im\phi)\exp(iEt).
\end{equation}
In order
to find numerically exact toroidal solitons (or vortices), we use imaginary time evolution starting from tori computed from a variational ansatz. 
As for the groundstate soliton in Sec.~\ref{gs_num}, we find that adding $1\%$ to the mass $M$ leads to
collapse of the toroidal soliton, whereas subtracting $1\%$ leads to delocalization of the wavefunction, that cannot self-trap anymore [see Fig.~\ref{fig:dynamics}(a)].
However, the toroidal state is azimuthally unstable.
Adding a small amplitude noise is enough to trigger the instability and due to the interplay between azimuthal phase ramp and amplitude modulation along the ring the whole structure starts to rotate~\cite{Desyatnikov:05}. Upon further propagation, the ring breaks into two humps (groundstate solitons), that eventually collapse individually [see Figs.~\ref{fig:dynamics}(b)-(f)]. This secondary collapse is possible because the mass of the torus fulfills $M_{\rm Torus}>2\times M_{\rm GS}$.

 \begin{figure}
 \centerline{\includegraphics[width=0.8\textwidth]{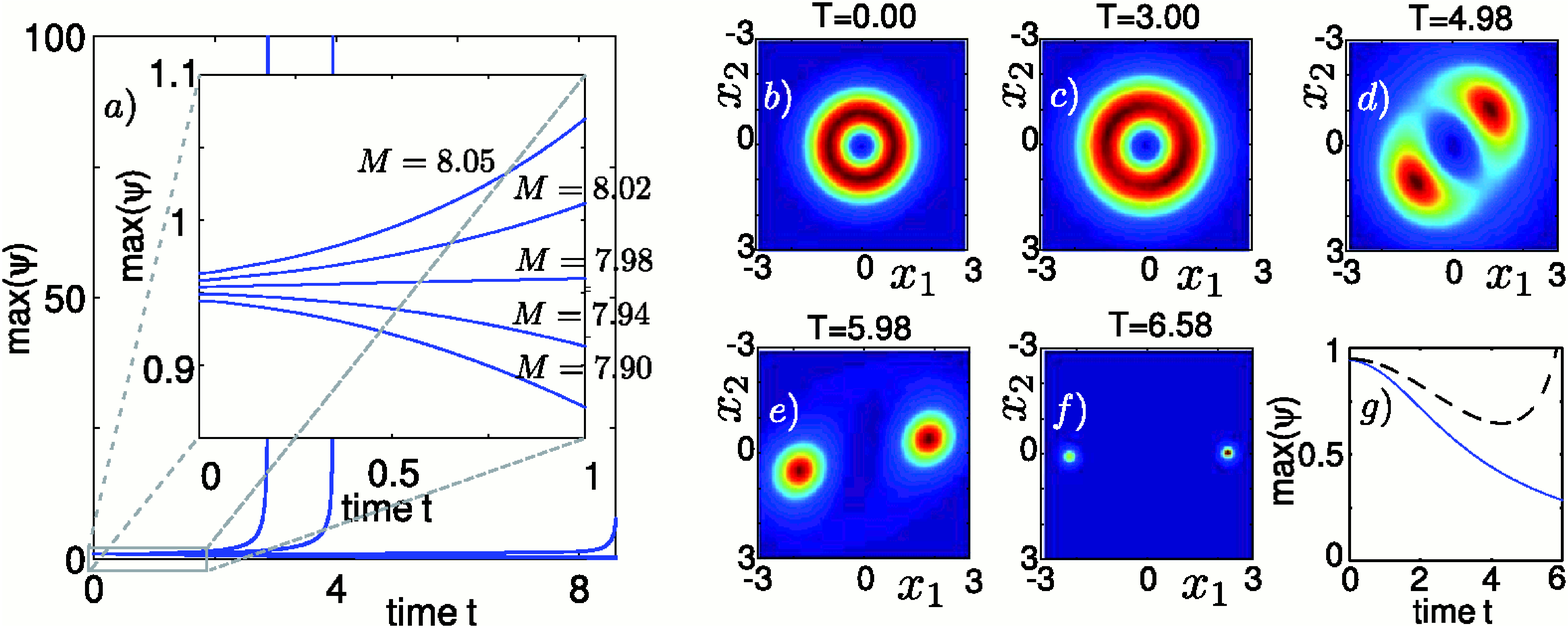}}
 \caption{Time evolution scenarios for the single charged ($m=1$) toroidal soliton.
a) Time evolution of the amplitude. When reducing the mass by $1\%$, the solution spreads, whereas it collapses for masses  $M\geq 7.98$. b)-f) The vortex may decay into two groundstate solitons due to an azimuthal instability, which then collapse individually.
g) The azimuthal instability can be triggered by adding random noise (here $1\%$ in amplitude), so that instead of spreading (blue solid line) a torus with $M=7.90$ decays and collapses (black dashed line). The corresponding 3D evolution of the modulus square of the wave function in the $(x_1,x_2)$ plane ($x_3=0$) is depicted in b)-f).
}
\label{fig:dynamics}
 \end{figure}

\section{Conclusions}
The phenomenon of collapse in  the nonlocal nonlinear Schr{\"o}dinger equation has been investigated in arbitrary dimensions. We showed that nonsingular nonlocal kernels do not support collapsing solutions. More general, it has been shown that for singular kernels $\sim 1/r^{\alpha}$, the necessary condition for collapse to occur is $\alpha \geq 2$.
Various collapse scenarios for the $1/r^2$ kernel have been studied numerically in  the physically relevant case of three transverse dimensions ($n=3$).
The critical time (or distance) when collapse occurs has been estimated by using self-similar solutions.
Apart from the groundstate solution, collapse of a toroidal state with absence and presence of local repulsive interaction has also been studied in this particular system.
Generally, it appears that the initial torus is azimuthally unstable. 
Finally, we showed that independently of the initial state, collapse can be arrested by adding a local repulsive contribution to the attractive $1/r^2$ kernel.

\section*{Acknowledgements}
The authors (S.~S.) would like to thank S.~K.~Turitsyn for fruitfull discussions
and the anonymous referees for valuable comments which helped to improve the paper.
Numerical simulations were partly performed on the SGI
XE Cluster and the Sun Constellation VAYU Cluster of the
Australian Partnership forAdvanced Computing (APAC). This
research was supported by the Australian Research Council.

\appendix
\section{Notations\label{notations}}
We use common functional analysis definitions.
The symbol $L_p(\mathbf{R}^n)$ denotes the space of measurable functions for which the Lebesgue measure
\begin{equation}
 \| \psi \|_p := \left (\int |\psi|^p \mathrm{d}x \right)^{1/p}
\end{equation}
exists. A function is in $L_\infty(\mathbf{R}^n)$, if there is a constant $C>0$, such that $|\psi|$ is smaller than $C$ for almost every
$x\in\mathbf{R}^n$.
The symbol $H^1(\mathbf{R}^n)$ (Sobolev space) is the space of measurable functions for which the norm, defined as
\begin{equation}
 \| \psi \|_{H^1} := \|\nabla \psi \|_2 + \|\psi\|_2
\end{equation}
is smaller than $\infty$. Here, $\int$ denotes $\int_{-\infty}^\infty$.

The Fourier transform $\tilde{f}$ of a function $f$ is given by
\begin{equation}
\tilde{f}(k) = \int e^{-2i\pi (k,x) } f(x) \mathrm{d}x
\end{equation}
with $x,k\in\mathbf{R}^n$, and $(\cdot,\cdot)$ denotes the scalar product.

Apart from that, $|x|=\sqrt{x_1^2 + x_2^2 + \dots} = r$ is the usual Euclidean distance.

\section{Hardy-Littlewood-Sobolev inequality\label{hardy_littlewood}}
Let $p,l>1$ and $0<\alpha<n$ with $1/p + \alpha/n + 1/l = 2$. Let $f\in L^p(\mathbf{R}^n)$ and $h\in L^l(\mathbf{R}^n)$. Then there exists a sharp constant $C(n,\alpha,p)$
independent of $f$ and $h$, such that
\begin{equation}
\left| \int \int \frac{f(x) h(y)}{ \left|x-y\right|^{\alpha}} \mathrm{d}x\mathrm{d}y\right| \leq C(n,\alpha,p) \|f\|_p \|h\|_l.
\label{eq:hardy_littlewood}
\end{equation}
The explicit functional dependence of the $C(n, \alpha,p)$ is known (see e.g. \cite{Lieb:2001}). For our purposes,  the special case $p=l=2n/(2n-\alpha)$ is important, where the sharp constant is given by
\begin{equation}
C(n,\alpha,p)=C(n,\alpha)=\pi^{\alpha/2}\frac{\Gamma(n/2-\alpha/2)}{\Gamma(n-\alpha/2)}\left[\frac{\Gamma(n/2)}{\Gamma(n)} \right]^{-1+\alpha/n}.
\label{eq:hardy_littlewood_constant}
\end{equation}
In this case there is equality in (\ref{eq:hardy_littlewood}) if and only if $h=Bf$ and
\begin{equation}
f(x)=A\left(\gamma^2+\left|x-a\right|^2\right)^{-(2n-\alpha)/2}.
\label{eq:hardy_littlewood_equation}
\end{equation}
for some complex constants $A$ and $B$, $0\neq\gamma$ real and $a\in\mathbf{R}^n$.~\cite{Lieb:2001}

\section{Gagliardo-Nirenberg-Sobolev inequality\label{gagliardo_nirenberg}}
Let $f\in H^1(\mathbf{R}^n)$ and $0<\sigma<2/(n-2)$. Then there exists a sharp constant $K(\sigma,n)$
independent of $f$, such that
\begin{equation}\label{eq:GN}
\int|f(x)|^{2\sigma+2}\mathrm{d}x \leq K(\sigma,n) \left(\int |\nabla f(x)|^2 \mathrm{d}x\right)^{\frac{\sigma n}{2}} \left(\int|f(x)|^{2}\mathrm{d}x\right)^\frac{2+\sigma (2-n)}{2}.
\end{equation}
The explicit functional dependence of the $K(\sigma,n)$ is known (see e.g. \cite{Sulem:1999}), and reads
\begin{equation}
K(\sigma,n)=\left(\sigma+1\right)\frac{2\left(2+2\sigma-\sigma n\right)^{-1+\sigma n/2}}{\left(\sigma n\right)^{\sigma n/2}}\frac{1}{\|R\|_2^{2\sigma}},
\label{eq:GN_constant}
\end{equation}
where $R$ is the positive solution of $\Delta R-R+R^{2\sigma+1}=0$.

\section{Numerical implementation of the singularity of $1/r^2$ }\label{sec:resolve_singularity}

Since we study collapse phenomena, it is crucial to define the properly in a numerical sense kernel $1/r^2$ and, in particular, its value at $r=0$ . To this end we 
used the following procedure (see, e.g., \cite{Martyna:JOCP:1999}). Firstly, we decompose the singular
kernel into a short range, singular at  the origin, and a nonsingular, long-range contributions.
\begin{equation}
 \frac{1}{r^2} = \frac{e^{-r^2/w_0^2}}{r^2} +  \frac{1-e^{-r^2/w_0^2}}{r^2}
\end{equation}
The long-range part can be used directly as it is in a real space, whereas the short-range part ($\frac{e^{-r^2/w_0^2}}{r^2}$)
 will be treated in the Fourier domain.
In our numerics, the parameter $w_0$ was chosen to be $w_0=7 \Delta x/2\pi$, with $\Delta x$ the step size of the spatial grid.
In Fourier domain, the product of two functions becomes a convolution, that is in particular also well
defined for certain singular functions such as $1/r^\alpha$ with $\alpha<n$, if the other function is ``well behaved``.
Hence, in Fourier domain the short ranged contribution can be
easily calculated, and the limit $|k|\rightarrow 0$ is then well defined.
\section*{References}
\bibliographystyle{unsrt}
\bibliography{collapse}

\end{document}